\documentclass{IEEEtran4PSCC}
\ifCLASSINFOpdf
   \usepackage[pdftex]{graphicx}
\else
   \usepackage[dvips]{graphicx}
\fi
\usepackage{algorithm,algorithmic}
\usepackage{bm}
\usepackage[cmex10]{amsmath}
\usepackage{amsthm,amssymb}
\usepackage{color}
\usepackage{booktabs}
\usepackage{tikz}
\newtheorem{theorem}{Theorem}

\newtheorem{lemma}[theorem]{Lemma}

\newtheorem{assumption}{Assumption}

\usepackage{graphicx}

\usepackage{flushend}

\makeatletter
\let\old@ps@headings\ps@headings
\let\old@ps@IEEEtitlepagestyle\ps@IEEEtitlepagestyle
\def\psccfooter#1{%
    \def\ps@headings{%
        \old@ps@headings%
        \def\@oddfoot{\strut\hfill#1\hfill\strut}%
        \def\@evenfoot{\strut\hfill#1\hfill\strut}%
    }%
    \def\ps@IEEEtitlepagestyle{%
        \old@ps@IEEEtitlepagestyle%
        \def\@oddfoot{\strut\hfill#1\hfill\strut}%
        \def\@evenfoot{\strut\hfill#1\hfill\strut}%
    }%
    \ps@headings%
}
\makeatother

\psccfooter{%
        \parbox{\textwidth}{\hrulefill \\ \small{22nd Power Systems Computation Conference} \hfill \begin{minipage}{0.2\textwidth}\centering \vspace*{4pt} \includegraphics[scale=0.06]{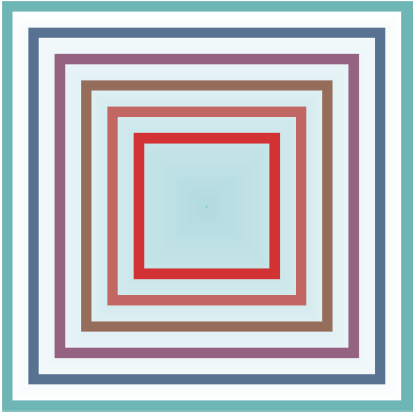}\\\small{PSCC 2022} \end{minipage} \hfill \small{Porto, Portugal --- June 27 -- July 1, 2022}}%
}

\newcommand{\cE}{{\cal E}}
\newcommand{\cG}{{\cal G}}
\newcommand{\cI}{{\cal I}}
\newcommand{\cM}{{\cal M}}
\newcommand{\cN}{{\cal N}}
\newcommand{\cV}{{\cal V}}
\newcommand{\bY}{{\bf Y}}
\newcommand{\dtoprule}{\specialrule{0.6pt}{0pt}{0.4pt}%
            \specialrule{0.6pt}{0pt}{\belowrulesep}%
            }
\newcommand{\dbottomrule}{\specialrule{0.6pt}{0pt}{0.4pt}%
            \specialrule{0.6pt}{0pt}{\belowrulesep}%
            }
\newcommand{\Cross}{\mathbin{\tikz [x=1.4ex,y=1.4ex,line width=.2ex] \draw (0,0) -- (1,1) (0,1) -- (1,0);}}%
            
\allowdisplaybreaks[4]

\begin{document}
%
\title{Online Distribution System State Estimation\\via Stochastic Gradient Algorithm}

\author{
\IEEEauthorblockN{Jianqiao Huang}

\IEEEauthorblockA{Department of Electrical and Computer Engineering \\
Illinois Institute of Technology\\
Chicago, USA\\
jhuang54@hawk.iit.edu}
\and
\IEEEauthorblockN{Xinyang Zhou\\Bai Cui}
\IEEEauthorblockA{Power System Engineering Department\\
National
Renewable Energy Laboratory\\
Golden, USA\\
\{xinyang.zhou, bai.cui\}@nrel.gov}
}


\maketitle

\begin{abstract}
Distribution network operation is becoming more challenging because of the growing integration of intermittent and volatile distributed energy resources (DERs). This motivates the development of new distribution system state estimation (DSSE) paradigms that can operate at fast timescale based on real-time data stream of asynchronous measurements enabled by modern information and communications technology. To solve the real-time DSSE with asynchronous measurements effectively and accurately, this paper formulates a weighted least squares DSSE problem and proposes an online stochastic gradient algorithm to solve it. The performance of the proposed scheme is analytically guaranteed and is numerically corroborated with realistic data on IEEE-123 bus feeder.
\end{abstract}


\thanksto{\noindent J. Huang is with the Department of Electrical and Computer Engineering, Illinois Institute of Technology, Chicago, USA (email: jhuang54@hawk.iit.edu).\\
X. Zhou and B. Cui are with Power Systems Engineering Center, National Renewable Energy Laboratory, Golde, CO 80401, USA (email: \{xinyang.zhou, bai.cui\}@nrel.gov).\\
This work was authored in part by the National Renewable Energy Laboratory, operated by Alliance for Sustainable Energy, LLC, for the U.S. Department of Energy (DOE). The views expressed in the article do not necessarily represent the views of the DOE or the U.S. Government. The U.S. Government retains and the publisher, by accepting the article for publication, acknowledges that the U.S. Government retains a nonexclusive, paid-up, irrevocable, worldwide license to publish or reproduce the published form of this work, or allow others to do so, for U.S. Government purposes.}

\section{Introduction}
The operation of distribution system is becoming challenging with the growing integration of distributed energy resources (DERs). DERs are introducing more intermittency and volatility into distribution networks, resulting in fast changing states. To maintain normal operations in such scenarios, numerous advanced optimization and control algorithms have been designed for the distribution system. However, what has been overlooked is distribution system state estimation (DSSE), without which most of the advanced control strategies cannot be properly deployed in practice.


DSSE estimates the states of the distribution network by utilizing network parameters, topology and meter reading. Different from transmission system state estimation (TSSE), DSSE features untransposed multiphase lines of high $r/x$ ratios, unbalaned multiphase loads, much less accurate and asynchronous measurement devices, but more volatile system states due to high penetration of DERs \cite{primadianto2016review}. Therefore, the single phase equivalent model and conventional solution methods used for TSSE cannot produce the most efficient and accurate results for DSSE. 

There have been a lot of efforts devoted to DSSE. Among the existing work, the most widely used method is to formulate a weighted least squares (WLS) problem for DSSE. Based on the choice of state variables, the WLS-based DSSE methods can be categorized into nodal voltage based, branch current based and power injection based methods. A three-phase DSSE method was early proposed with bus voltage as state variables in \cite{Baran1994}, and developments have been made ever since. The authors in \cite{Lu1995} propose to convert nodal power injection measurements to the current injection measurements. To increase the computation efficiency and accuracy, a decomposition method and a load allocation approach is proposed in \cite{Deng2002} by exploiting the network topology. An augmented nodal analysis formulation is proposed in \cite{Francis2013} to increase numerical robustness. A semi-definite programming approach is proposed in \cite{Yao2019} to model virtual-measurements as equality constraints. The first branch-current based method is proposed in \cite{Baran1995} to reduce computation burden and condition number of the gain matrix. The voltage magnitude measurements are incorporated into the branch-current based approach in \cite{Baran2009}. This method is further enhanced by considering the synchrophasor measurements in \cite{Marco13}. A numerically robust and scalable power injection based method is proposed in \cite{zhou2020} to further increase computation efficiency. Dynamic system state estimation approaches are proposed in \cite{Zhao2017, Filho2009, Song2020} for handling fast time-varying system.

However, most existing works assume all the measurement data to be available in time, while in real-time DSSE, data from the heterogeneous and asynchronous measurement devices \cite{Zheng2013,von2014}, like PMU and smart meters, may have largely different sampling rates \cite{Angioni2013,Alimardani2015}. Furthermore, communication loss and delay, as well as bad data discarding, also lead to missing and asynchronicity of measurement data. The resulting DSSE problem is not observable, and the classic DSSE methods that assume synchronous data can not be applied. To solve this issue, authors in \cite{Ali2015} model load variation between two consecutive samples as a random variable following a normal distribution, and decrease measurement weights as time goes by. The authors in \cite{BS2014} assume synchronization errors to follow normal distribution with zero mean so that they can be modeled as measurement errors and states can be estimated via WLS method. In \cite{Guido2019}, outdated estimated states are used as measurements to ensure observability.

In this work, we do not make any assumption on the distribution of the load variation or synchronization error. We will formulate the DSSE as a time-varying weighted least squares (WLS) problem and propose to solve it via a stochastic gradient descent algorithm instead of Gauss-Newton (GN) method. The stochastic gradient algorithm can accommodate for asynchronous stream of measurement data and deliver accurate estimation efficiently. Stability of the proposed algorithm theoretically proved and numerically illustrated.

The remainder of this paper is organized as follows. In Section~\ref{system model}, we introduce the notation, distribution system modeling, and the formulation of time-varying WLS. We next propose the stochastic gradient descent (SGD) algorithm for solving the real-time DSSE problem and provide the convergence analysis in Section~\ref{SGD}. In Section~\ref{Simulation}, the proposed algorithm is numerically tested on IEEE-123 bus feeder. We conclude the paper in Section~\ref{conclusion}.

\section{Distribution System State Estimation}\label{system model}
\subsection{Notations}
In this paper, we use bold letters to represent matrices, e.g., $\mathbf{A}$, italic bold letters to represent vectors, e.g, $\mathbf{A}$ and $\bm{a}$, and non-bold letters to represent {scalars}, e.g., $A$ and $a$. For matrix $\mathbf{A}$, $\mathbf{A}^{\top}$, $\mathbf{A}^{-1}$ and $\tau (\mathbf{A})$ denote its transpose, inverse and smallest eigenvalue, respectively. $\mathbf{diag}(\bm{a})$ denotes a diagonal matrix with vector $\bm{a}$ as its diagonal. $\mathbb{E}[x]$ denotes the expected value of random variable $x$. $\mathfrak{i}:=\sqrt{-1}$ is used as the imaginary unit.

\subsection{Distribution system}
We consider a multiphase distribution system denoted by a graph $\cG=\{\cV_0, \cE\}$, where $\cV_0 = \{0\}\cup\cV=\{0\}\cup\{1, 2, ..., N\}$ with bus 0 being the slack bus and the set of $\cV$ collecting the rest buses, and $\cE = \{ (i, j) \subseteq \cV_0 \times \cV_0 \}$ denotes the set of lines between buses. Each bus has up to three available phases. We define each phase of a bus in $\cV$ as a node, and define the set of all the nodes as $\cN:= \{1, ..., n \}$ where $n$ is the total number of nodes in the system. Denote by $\cN^l \subseteq \cN$ the subset of all load nodes.

Similar to the buses, a line may have up to three available phases, each defined as a branch. A line with $l$ phases features $l$-by-$l$ admittance and impedance matrices. The admittance and impedance matrix for the line between bus $i$ and $j$ are denoted by $\bm{y}_{ij}$ and $\bm{z}_{ij}=\bm{r}_{ij}+\mathfrak{i}\bm{x}_{ij}$, respectively. Let $\cV_{i}$ denote the set of adjacent buses of bus~$i$ and $\cI_i$ the set of nodes within bus $i$. Then the bus admittance matrix $\bY$ can be characterized by its submatrix $\bY({\cI_i,\cI_j})$ corresponding to any $(i,j)\in\cE$ as:
\begin{align}
\bY({\cI_i,\cI_j})
=\begin{cases}
    \bm{0}, &\text{if } j \neq i, \text{and } j \not \in \cV_i, \nonumber\\
    -\bm{y}_{ij}, &\text{if } j \in  \cV_i, \nonumber\\ 
    \sum_{k \in \cV_i} \bm{y}_{ik}, & i=j. 
\end{cases}
\end{align}
Also see \cite{Bazrafshan18} for more details of multiphase distribution system modeling.


\subsection{Formulation of DSSE}
In DSSE, the measurement model for a distribution system is expressed as follows,
\begin{align}
    \bm{y}=\bm{h}(\bm{z})+\bm{e},
\end{align}
where $\bm{y} \in \mathbb{R}^{m}$, $\bm{z} \in  \mathbb{R}^{2n}$, $\bm{e}\in \mathbb{R}^{m}$ denote the vector of measurements, state variables. 

\subsubsection{Measurements}
For a distribution system, the measurements usually include virtual-measurements, real-time measurements and pseudo-measurements. The virtual-measurements are referred to the known knowledge of zero-injection nodes without any load or generation. The real-time measurements are sampled by a limited number of phasor measurement units and smart meters. Both virtual-measurements and real-time measurements are assumed to have very small standard deviation \cite{Monticelli85}. Meanwhile, most of the load nodes in the distribution system do not have real-time measurements. To have an observable real-time DSSE problem, pseudo-measurements for load nodes are estimated from historical load data with very large standard deviation \cite{Ghosh1997}. 

\subsubsection{System States}
System states are the set of parameters that can be used to determine all system parameters. In power system, the state variable $\bm{z}$ can be nodal complex voltage \cite{Baran1994}, branch complex current \cite{Baran1995} or nodal complex power injection \cite{zhou2020}. Accordingly, we have nodal-voltage based, branch-current based and power injection based DSSE \cite{Prima2017}. Both the nodal-voltage based and branch-current based DSSE can be formulated in polar and rectangular coordinate, while the power injection based DSSE is formulated in rectangular coordinate. The methods to solve the three categories of DSSE are as follow. Nodal-voltage-based DSSE directly applies the WLS method. Branch-current-based DSSE computes the equivalent current injection and updates residuals before applying WLS method, and updates nodal complex voltage through a forward sweep calculation after the WLS method in each iteration. See \cite{zhou2020} for the details of the algorithm of power injection based DSSE.

\subsubsection{Weighted Least Squares}
Weighted least squares (WLS) estimator is the most popular one among the existing DSSE methods. Based on the assumption that the measurement errors are independent random variables of normal distribution, WLS estimator can achieve the maximized joint probability distribution, or the weighted sum of squared residuals:
\begin{eqnarray}\label{eq:WLS}
    &\underset{\bm{z}}{\min}& \frac{1}{2}(\bm{y}-\bm{h}(\bm{z}))^{\top}\mathbf{W} (\bm{y}-\bm{h}(\bm{z})),
\end{eqnarray}
where $\mathbf{W}=\text{diag} \{1/\sigma_1^2,\ldots,1/\sigma_m^2\}\in\mathbb{R}^{m\times m}$ is a diagonal weight matrix with the diagonal elements being the reciprocal of the variance of the corresponding measurement errors, reflecting the accuracy of the associated meters. Therefore, the real-time measurements and virtual measurements are assigned with small variance, while the pseudo-measurements are assigned with large variance. We therefore have the following assumption on observability of \eqref{eq:WLS}.

\begin{assumption}\label{observable}
The distribution system is fully observable.
\end{assumption}
Full observability of the distribution system can be achieved by using the power injection pseudo-measurements and virtual measurements.

\subsection{Conventional Gauss-Newton Algorithm}
To solve for WLS problem~\eqref{eq:WLS}, iterative Gauss-Newton (G-N) algorithm is usually applied based on solving its first-order optimality condition:
\begin{align}
    -\mathbf{H}^{\top}\mathbf{W}\big(\bm{y}-\bm{h}(\bm{z})\big)=\bm{0}_{2n}, \label{FONC}
\end{align}
where $\mathbf{H} \in \mathbb{R}^{m \times 2n}$ denotes the Jacobian matrix of $\bm{h}(\bm{z})$ with respect to $\bm{z}$. Denote the state at the iteration $k$ by $\bm{z}(k)$ and the first-order Taylor series approximation of \eqref{FONC} around the state $\bm{z}(k)$ yields:
\begin{eqnarray}\label{GN algorithm}
    \bm{z} (k+1)=\bm{z} (k)+\mathbf{G}^{-1}\mathbf{H}^{\top} \mathbf{W} \Big(\bm{y}-\bm{h}\big(\bm{z} (k)\big)\Big)    
\end{eqnarray}
where $\mathbf{G}=\mathbf{H}^{\top}\mathbf{W}\mathbf{H}$
is defined as the gain matrix. $\bm{z} (k)$ is the solution vector at the $k$th iteration. \eqref{GN algorithm} is known as the iterative Gauss-Newton (G-N) algorithm. The algorithm is repeated until the stopping criterion is satisfied.

\subsection{Motivation for Real-Time DSSE Development}
We have so far discussed about the state estimation based on a single batch of the measurements, i.e., offline DSSE. In real-time operation, however, control center may receive temporal series of snapshots of measurements. To track the system states in real-time, the WLS problem \eqref{eq:WLS} is next reformulated for any time $t$ as:
\begin{eqnarray}\label{eq:real-time WLS}
    &\underset{\bm{z}_t}{\min}& \frac{1}{2}(\bm{y}_t-\bm{h}(\bm{z}_t))^{\top}\mathbf{W}_t (\bm{y}_t-\bm{h}(\bm{z}_t)),
\end{eqnarray}
where $\bm{y}_t \in \mathbb{R}^{m}$, $\bm{z}_t \in  \mathbb{R}^{2n}$, $\mathbf{W}_t\in \mathbb{R}^{m}$ denote the vector of measurements, state variables and weight matrix at time $t$.
Similarly, we can implement the G-N algorithm to the real-time WLS problem for each time $t$. 


However in practice, when faster estimation is desired, we may not be able to finish computing all iterations at time $t$ given limited computation capability. More importantly, we may not obtain the whole batch of measurement $\bm{y}_t$ due to asynchronous measurement and communication delays, rendering the G-N updates slow or even infeasible. To address this issue, we will propose a new real-time DSSE paradigm based on stochastic gradient algorithm.




\section{Solve Real-Time DSSE via Stochastic Gradient Descent}\label{SGD}

One way to significantly improve the timeliness of the DSSE solvers is to take advantage of the asynchronous measurement data, instead of waiting for the whole batch of measurement data to start computation. However, the asynchronous measurements from heterogeneous measurement devices pose a challenge to the G-N algorithm, because the number of real-time measurements at each interval is less than the number of state variables and thus the G-N algorithm cannot be properly executed. To solve this issue, we propose a real-time SGD algorithm, which does not require the whole batch of measurement data for computation at each iteration, and guarantees provable performance.

\subsection{Online Gradient Algorithm}
We first propose a gradient algorithm which runs sufficient iterations until convergence for each $t$ to completely solve \eqref{eq:real-time WLS}:
\begin{align}\label{eq: grad_LPF+Feedback}
    \bm{z}_t(k+1)=\bm{z}_{t}(k) - \epsilon \mathbf{H}^{\top} \mathbf{W}_t (\bm{h}(\bm{z}_{t}(k))-\bm{y}_t),
\end{align}
where $\epsilon$ is a constant stepsize and $\mathbf{H}^{\top} \mathbf{W}_t (\bm{h}(\bm{z}_{t}(k))-\bm{y}_t)$ is the gradient of \eqref{eq:real-time WLS} evaluated at $\bm{z}_{t}(k)$. Such gradient updates can be shown to converge to the solutions of \eqref{eq:real-time WLS} given small enough stepsize $\epsilon$ \cite{zhou2020}. 

In scenarios where we may only have time to run limited iterations from $t$ to $t+1$, e.g., 1 iteration for each $t$, we can hardly expect \eqref{eq: grad_LPF+Feedback} to converge in time. Therefore, we propose the following online gradient descent method for time slot $t$ with a constant stepsize $\epsilon$:
\begin{align}\label{eq: r-t grad_LPF+Feedback}
    \bm{z}_{t} = \bm{z}_{t-1} - \epsilon \mathbf{H}^{\top} \mathbf{W}_t (\bm{h}(\bm{z}_{t-1})-\bm{y}_{t}).
\end{align}
Note that at each time $t$, the algorithm is initialized at $\bm{z}_{t-1}$ and the gradient is calculated based on the current measurement $\bm{y}_{t}$. We summarize the online gradient-based DSSE (GD) algorithm in Algorithm~\ref{alg:real-time GD}. 

\begin{algorithm}[h]
	\caption{Online Gradient Algorithm for DSSE} 
	\begin{algorithmic}\label{alg:real-time GD}	
        \STATE Estimator runs in real time with a constant stepsize $\eta$:\\
		\REPEAT
        \STATE 1)
	    Receive a whole batch of measurements $\bm{y}_t$.\\
	    \STATE 2)
	    Update weight matrix $\mathbf{W}_t$ and evaluates measurement function $\bm{h}(\bm{z}_{t-1})$.
	    \STATE 3) Update system states with gradient step \eqref{eq: r-t grad_LPF+Feedback}.
        \STATE 4)
        Move on to the next time step $t\leftarrow t+1$.\\
        \UNTIL being ended.
	\end{algorithmic}
\end{algorithm}

\subsection{Stochastic Gradient Algorithm}\label{propose SGA}

The scenario used in Algorithm~\ref{alg:real-time GD} is rather ideal, namely, we have assumed that a whole batch of measurement data can be obtained for gradient update at each time $t$. In practice, however, asynchronous measurement and communication delay and loss prevent us from collecting the entire batch of measurement in time.
Assume that we only collect $m_t\leq m$ measurements at time $t$, rendering the standard gradient updates \eqref{eq: r-t grad_LPF+Feedback} incomplete. In order to avoid waiting for the entire batch of measurement to arrive and to provide timely updates of the system states, we propose to use stochastic gradient to update with available incomplete measurement data as follows:
\begin{align}\label{eq:z_update_LPF_SGD}
    \bm{z}_t=\bm{z}_{t-1}-\epsilon \mathbf{H}^s_t(\bm{z}_{t-1})^{\top} \mathbf{W}^s_t (\bm{h}^s_{t}(\bm{z}_{t-1})-\bm{y}_t),
\end{align}
where $\bm{y}_t \in\mathbb{R}^{m_t}$ is the available measurement vector at time $t$, $\mathbf{W}^s_t \in \mathbb{R}^{m_t \times m_t}$ is the submatrix of the weight matrix corresponding to the available measurement, 
$\bm{h}^s_t(\bm{z}_{t-1}):\mathbb{R}^{2n}\rightarrow \mathbb{R}^{m_t}$ and  $\mathbf{H}^s_t(\bm{z}_{t-1}) \in \mathbb{R}^{m_t \times 2n}$ denote the partial measurement functions and the measurement Jacobian matrix corresponding to the available measurements at time $t$ evaluated at $\bm{z}_{t-1}$. 

We execute dynamic~\eqref{eq:z_update_LPF_SGD} in real time to generate Algorithm~\ref{alg:r-t SGD}. The algorithm repeats the following two steps after the arrival of $m_t$ measurements. First, the estimator updates measurement Jacobian matrix, weight matrix and evaluates measurement residuals. Second, estimator updates state variables with one SGD step.


\begin{algorithm}[h]
	\caption{Online SGD Algorithm for DSSE} 
	\begin{algorithmic}\label{alg:r-t SGD}	
        \STATE Estimator runs in real-time with a constant stepsize $\eta$:\\
		\REPEAT		
        \STATE 1)
	    Receive a batch of measurements $\bm{y}_t$.\\
	    \STATE 2)
	    Updates Jacobian matrix $\mathbf{H}_t^s $, weight matrix $\mathbf{W}_t^s$ and evaluates measurement function $\bm{h}^s_t(\bm{z}_{t-1})$.
	    \STATE 3)
	     Update system states with stochastic gradient step \eqref{eq:z_update_LPF_SGD}
        \STATE 4)
        Move on to the next time step $t\leftarrow t+1$.\\
        \UNTIL being ended.
	\end{algorithmic}
\end{algorithm}


\subsection{Convergence Analysis}\label{convergence analysis}
In this subsection, we will show the convergence of SGD algorithm~\eqref{eq:z_update_LPF_SGD} with a fixed stepsize. Because the (stochastic) gradient used in \eqref{eq:z_update_LPF_SGD} corresponds to a nonconvex optimization problem \eqref{eq:real-time WLS} whose optimality and gradient dynamics are difficult to characterize, for convergence analysis we use the convex counterpart of \eqref{eq:real-time WLS} by replacing the nonlinear function $\bm{h}$ with its linearization $\mathbf{H}$ as follows:
\begin{eqnarray}\label{eq:real-time convex WLS}
    &\underset{\bm{z}_t}{\min}& \frac{1}{2}(\bm{y}_t-\mathbf{H} \cdot \bm{z}_t)^{\top}\mathbf{W}_t (\bm{y}_t-\mathbf{H} \cdot \bm{z}_t).
\end{eqnarray}
The gradient dynamics for solving \eqref{eq:real-time convex WLS} is written as:
\begin{align}\label{eq:z_update_LPF}
    \bm{z}_t=\bm{z}_{t-1}-\epsilon \mathbf{H}^{\top} \mathbf{W}_t (\mathbf{H}\cdot\bm{z}_{t-1}-\bm{y}_t),
\end{align}
which can be proved to converge asympototically to the solution of \eqref{eq:real-time convex WLS}.

For ease of expression, we use the notations defined in the following table in our proof.

\begin{table}[htbp]
\normalsize
\begin{center}
\caption{Nomenclature}
\begin{tabular}{l l}
\cline{1-2} 
\hline
$f(\bm{z}_t)$ & gradient of \eqref{eq:real-time convex WLS}: $\mathbf{H}^{\top} \mathbf{W}_t (\mathbf{H}\cdot\bm{z}_{t-1}-\bm{y}_t)$;\\
$\tilde{f}(\bm{z}_t)$ & gradient in \eqref{eq: r-t grad_LPF+Feedback}: $\mathbf{H}^{\top} \mathbf{W}_t (\bm{h}(\bm{z}_{t-1})-\bm{y}_{t})$;\\
$f_t(\bm{z}_t;\bm{\xi}_t)$ & stochastic gradient of \eqref{eq:real-time convex WLS}:\\ &$\mathbf{H}^s_t(\bm{z}_{t-1})^{\top} \mathbf{W}^s_t (\mathbf{H}^s_t \cdot\bm{z}_{t-1}-\bm{y}_t)$;\\
$\tilde{f}_t(\bm{z}_t;\bm{\xi}_t)$ & stochastic gradient in \eqref{eq:z_update_LPF_SGD}:\\ &$\mathbf{H}^s_t(\bm{z}_{t-1})^{\top} \mathbf{W}^s_t (\bm{h}^s_{t}(\bm{z}_{t-1})-\bm{y}_t)$;\\
$\bm{z}^*_t$ &optimal solution of \eqref{eq:real-time convex WLS};\\
$\tilde{\bm{z}}^*_t$ & optimal solution of \eqref{eq:real-time WLS}.\\
\hline
\end{tabular}
\label{NOMENCLATURE}
\end{center}
\end{table}





We proceed with the following reasonable assumptions for analytical characterization.

\begin{assumption}\label{A1}
The discrepancy between the gradient of \eqref{eq:real-time convex WLS} and the gradient in \eqref{eq: r-t grad_LPF+Feedback} is bounded for any feasible $\bm{z}$. Consequently, we have
    \begin{eqnarray}
        \|\tilde{f}(\bm{z};\xi)-f(\bm{z};\xi)\|^2\leq \Delta_1,\ \forall t.\label{eq:delta2}
    \end{eqnarray}
    where $f(\bm{z};\xi)$ and $\tilde{f}(\bm{z};\xi)$ denote the stochastic gradients of \eqref{eq:real-time convex WLS} and \eqref{eq:real-time WLS}, respectively. $\xi$ denotes the randomness introduced by the SGD algorithm.
\end{assumption}

\begin{assumption}\label{A2}
The expected value of the squared $L2$ norm for the stochastic gradient is bounded for any feasible $\bm{z}$. As a result, we have
    \begin{eqnarray}
        \mathbb{E}[\|f(\bm{z};\xi)\|^2]\leq  \sigma_f^2.\label{eq:sigma_g}
    \end{eqnarray}
\end{assumption}

\begin{assumption}\label{A3}
The difference of the optimal solutions of \eqref{eq:real-time convex WLS} between two consecutive time slots is bounded,
\begin{eqnarray}
    &\|\bm{z}_{t+1}^*-\bm{z}_{t}^*\| \leq \Delta_{\bm{z}}, \forall t.\label{eq:delta1}     
\end{eqnarray}
\end{assumption}

\begin{assumption}\label{A4}
The discrepancy between $\tilde{\bm{z}}_t^*$, the optimal solution of \eqref{eq:real-time WLS}, and $\bm{z}_t^*$, the optimal solution of \eqref{eq:real-time convex WLS}, is very small.
\end{assumption}

Under the Assumption~\ref{observable}, the gain matrix $\mathbf{H}^{\top} \mathbf{W}_t \mathbf{H}$ is positive definite given the complete measurements. As the gain matrix is the Hessian matrix of \eqref{eq:real-time convex WLS}, the optimization problem \eqref{eq:real-time convex WLS} is strongly convex. Denotes by $\tau_t$ the smallest eigenvalue of the gain matrix $\mathbf{H}^{\top} \mathbf{W}_t \mathbf{H}$ for notational simplicity, and we have the following lemma.
\begin{lemma} For any feasible $\bm{z}_t$ and $\bm{z}'_t$, one has:
\begin{align}\label{eq:B1}
        &\big(f(\bm{z}_t)-f(\bm{z}'_t)\big)^{\top}(\bm{z}_t-\bm{z}'_t)
        \geq \tau_t \|\bm{z}_t-\bm{z}'_t\|^2\geq \tau_1 \|\bm{z}_t-\bm{z}'_t\|^2,
\end{align}
      where denote by $\tau_1$ the lower bound of the smallest eigenvalue of any gain matrix.
\end{lemma}

\begin{lemma}[Unbiased Estimation \cite{2018Bottou}]\label{lem:unbiased estimate}
Stochastic gradient $f(\bm{z};\xi_t)$ is an unbiased estimate of the full gradient $f(\bm{z})$ given $\{\xi_1, \xi_2, ..., \xi_{t-1}\}$, i.e.,
\begin{eqnarray}\label{eq:ue}
\mathbb{E}[f(\bm{z};\xi_t)|\xi_1, ..., \xi_{t-1}]=f(\bm{z}).
\end{eqnarray}
\end{lemma}

Denote by $\eta$ the constant stepsize of the SGD algorithm, and we have the following theorem.

\begin{theorem}[Convergence]\label{the:convergence}
Under the Assumptions~\ref{observable}--\ref{A4}, if the constant stepsize satisfies $0<\eta \leq\frac{1}{2 \tau_1}$, then the proposed algorithm achieves:
\begin{eqnarray}\label{eq:convergence}
    \lim_{t \rightarrow \infty}\mathbb{E}[\|\bm{z}_{t}^*-\bm{z}_t\|^2]= \frac{\eta^2  \sigma_f^2+\eta^2 \Delta_1+\Delta_{\bm{z}}}{2\eta \tau_1}, 
\end{eqnarray}
where $\bm{z}_{t}^*$ and $\bm{z}_{t}$ denote the optimal solution of \eqref{eq:real-time convex WLS} and the one obtained via the proposed online SGD algorithm at time $t$. 
\end{theorem}
\begin{proof}
From the SGD algorithm, we have
\begin{align}\label{derivation1:bound_FbSGDvsNoFbGD}
    &\|\bm{z}_{t+1}-\bm{z}_{t+1}^*\|^2\nonumber\\
    =&\|\bm{z}_{t+1}-\bm{z}_{t}^*-(\bm{z}_{t+1}^*-\bm{z}_{t}^*)\|^2\nonumber\\
    \leq&\|\bm{z}_t-\eta \tilde{f}(\bm{z}_t;\xi_t)-\bm{z}_{t}^*-(\bm{z}_{t+1}^*-\bm{z}_{t}^*)\|^2\nonumber\\
    \leq& \|\bm{z}_t-\eta \tilde{f}(\bm{z}_t;\xi_t)-\bm{z}_{t}^*\|^2+\|\bm{z}_{t+1}^*-\bm{z}_{t}^*\|^2\nonumber\\
    =&\|\bm{z}_t-\bm{z}_{t}^*-\eta f(\bm{z}_t;\xi_t)+\eta f(\bm{z}_t;\xi_t)-\eta \tilde{f}(\bm{z}_t;\xi_t)\|^2\nonumber\\
    &+\|\bm{z}_{t+1}^*-\bm{z}_{t}^*\|^2\nonumber\\
    \leq& \|\bm{z}_t-\bm{z}_{t}^*-\eta f(\bm{z}_t;\xi_t)\|^2\nonumber\\
    +&\eta^2 \| f(\bm{z}_t;\xi_t)- \tilde{f}(\bm{z}_t;\xi_t)\|^2+\|\bm{z}_{t+1}^*-\bm{z}_{t}^*\|^2\nonumber\\
    =&\|\bm{z}_t-\bm{z}_{t}^*\|^2+\eta^2 \| f(\bm{z}_t;\xi_t)\|^2\nonumber\\
    &-2\eta (\bm{z}_t-\bm{z}_{t}^*)^{\top} f(z_t;\xi_t)\nonumber\\
    &+\eta^2 \| f(\bm{z}_t;\xi_t)- \tilde{f}(\bm{z}_t;\xi_t)\|^2+\|\bm{z}_{t+1}^*-\bm{z}_{t}^*\|^2
\end{align}
The first inequality comes from the nonexpansiveness of projection operation, the second and the third from the triangle inequality of the norm. Taking expectation on both sides of \eqref{derivation1:bound_FbSGDvsNoFbGD} gives:
\begin{align}\label{eq: exp-1}
     &\mathbb{E}[\|\bm{z}_{t+1} -\bm{z}_{t+1}^*\|^2]\nonumber\\
     \leq& \mathbb{E}[\|\bm{z}_t-\bm{z}_{t}^*\|^2]+\eta^2 \mathbb{E}[\| f(\bm{z}_t;\xi_t)\|^2]\nonumber\\
    &-2\eta \mathbb{E}[(\bm{z}_t-\bm{z}_{t}^*)^{\top} f(z_t;\xi_t)]\nonumber\\
    &+\eta^2 \mathbb{E}[\| f(\bm{z}_t;\xi_t)- \tilde{f}(\bm{z}_t;\xi_t)\|^2]+\|\bm{z}_{t+1}^*-\bm{z}_{t}^*\|^2\nonumber\\
    =& \mathbb{E}[\|\bm{z}_t-\bm{z}_{t}^*\|^2]+\eta^2 \sigma_f^2-2\eta \mathbb{E}[(\bm{z}_t-\bm{z}_{t}^*)^{\top} f(z_t))]\nonumber\\
    &+\eta^2 \Delta_1^2+ \Delta_{\bm{z}}
\end{align}
where the first equality comes from \eqref{eq:sigma_g}, \eqref{eq:delta2}, \eqref{eq:delta1}. The third term of RHS of the first equality can be expressed by the unbiased estimate as follows,
\begin{align}\label{eq:unbais estimate}
 & 2\eta \mathbb{E}[(\bm{z}_t-\bm{z}_{t}^*)^{\top} f(z_t,\xi_t)]\nonumber\\
 =& 2\eta \mathbb{E}[\mathbb{E}[(\bm{z}_t-\bm{z}_{t}^*)^{\top} f(z_t,\xi_t)|\xi_1, \xi_2, ..., \xi_{t-1}]]\nonumber\\
 =& 2\eta \mathbb{E}[(\bm{z}_t-\bm{z}_{t}^*)^{\top} \mathbb{E}[f(z_t,\xi_t)|\xi_1, \xi_2, ..., \xi_{t-1}]]\nonumber\\
 =& 2\eta \mathbb{E}[(\bm{z}_t-\bm{z}_{t}^*)^{\top} f(z_t)]
\end{align}
We substitute RHS of \eqref{eq:unbais estimate} for $2\eta \mathbb{E}[(\bm{z}_t-\bm{z}_{t}^*)^{\top} f(z_t)]$ into \eqref{eq: exp-1} to obtain,
\begin{align}
 &\mathbb{E}[\|\bm{z}_{t+1} -\bm{z}_{t+1}^*\|^2]\nonumber\\
=& \mathbb{E}[\|\bm{z}_t-\bm{z}_{t}^*\|^2]+\eta^2 \sigma_f^2-2\eta \mathbb{E}[(\bm{z}_t-\bm{z}_{t}^*)^{\top} f(z_t))]\nonumber\\
    &+\eta^2 \Delta_1^2+ \Delta_{\bm{z}}\nonumber\\
\leq & \mathbb{E}[\|\bm{z}_t-\bm{z}_{t}^*\|^2]+\eta^2 \sigma_f^2-2\eta \tau_1 \mathbb{E}[\|\bm{z}_t-\bm{z}_{t}^*\|^2]\nonumber\\
    &+\eta^2 \Delta_1^2+ \Delta_{\bm{z}}\nonumber\\
    =& (1-2\eta \tau_1) \mathbb{E}[\|\bm{z}_1-\bm{z}_1^*\|^2]
    +(\eta^2 \sigma_f^2
    +\eta^2 \Delta_1+\Delta_{\bm{z}})\nonumber\\
    \leq& (1-2\eta \tau_1)^t \mathbb{E}[\|\bm{z}_1-\bm{z}_1^*\|^2]\nonumber\\
    &+\frac{1-(1-2\eta \tau_1)^t}{2\eta \tau_1} (\eta^2 \sigma_f^2
    +\eta^2 \Delta_1+\Delta_{\bm{z}})
\end{align}
where the first inequality comes from \eqref{eq:B1} and the second inequality is obtained by repeating all the previous steps for $t$ times. 

For $0<\eta\leq \frac{1}{2\tau_1}$, one has $0 \le 1-2\eta \tau_1<1$, so \eqref{eq:convergence} follows.

\end{proof}

\subsection{Power Injection Based DSSE}\label{algorithm implementation}
The modern distribution system will features a lot of zero injection nodes. The power injection measurements of these nodes are accurate, so their weights will be infinite or set very large values for the nodal-voltage based DSSE method. It brings in numerical issue. To circumvent such a problem, the branch-current based DSSE \cite{Baran1995} and power injection based DSSE \cite{zhou2020} are proposed. Based on the numerical study, both approaches achieve similar estimation accuracy, but the latter approach is more efficient. Therefore, we adopt the formulation of the latter approach in this paper. 

We select nodal power injection $[\bm{p}^{\top}, \bm{q}^{\top}]^{\top}$ for all nodes as the state variables $\bm{z}$. Denote by $\hat{v}_{i,t}$ voltage magnitude measurement for any node $i$ collected in the set $\cM^v_t$ and by $(\hat{p}_{i,t}, \hat{q}_{i,t})$ load pseudo-measurement for any node $i$ collected in the load set $\cN^l$. Then problem~\eqref{eq:real-time WLS} is reformulated as:
\begin{subequations}\label{eq:set}
	\begin{eqnarray}
	&\hspace{-9mm}\underset{\bm{p}_t,\bm{q}_t}{\min}&\hspace{-5mm}\sum_{i\in\cN^l}\hspace{-2mm}\Big(\frac{\hspace{-0.1mm}(p_{i,t}-\hat{p}_{i,t})^2}{2\sigma_{p_i}^2}\hspace{-1mm}+\hspace{-1mm}\frac{(q_{i,t}-\hat{q}_{i,t})^2}{2\sigma_{q_i}^2}\hspace{-0.1mm}\Big)\hspace{-1mm}+\hspace{-3mm}\sum_{i\in\cM^v}\hspace{-2mm}\frac{(v_{i,t}-\hat{v}_{i,t})^2}{2\sigma_{v_i}^2},\label{eq:WLSt}\\
	&\hspace{-9mm}\text{s.t.}& \hspace{-2mm}\bm{v}_t=\bm{g}(\bm{p}_t,\bm{q}_t),\label{eq:voltaget}\\
	&\hspace{-9mm}&\hspace{-2mm} (\bm{p}_t,\bm{q}_t)\in\bm{\Omega}_t,\label{eq:Omegat}
	\end{eqnarray}
\end{subequations}

where $\bm{g}$ denotes the nonlinear power flow equations, and $\bm{\Omega}_t =\Cross_{i\in\cN} {\Omega}_{i,t}$ with $\Omega_{i,t}=\mathbb{R}^2, \forall i\in\cN^l$ for pseudo-measurement and $\Omega_{i,t}=\{0,0\},\forall i\in\cN\backslash\cN^l$ for virtual measurement. Virtual measurements are zero power injections of the nodes without any load. These measurements are accurate, so we assign state variables of the nodes without any load zero.

The proposed stochastic gradient algorithm is implemented as follows:
\begin{subequations}\label{eq:SGD}
\begin{align}
&p_{i,t+1}=p_{i,t}-\epsilon\Big(\sum_{j\in\cM^v_t}\frac{\partial v_j}{\partial p_i} \frac{(v_{j,t}-\hat{v}_{j,t})}{\sigma^2_{v_j}}+\frac{\big(p_{i,t}-\hat{p}_{i,t}\big)}{\sigma^2_{p_i}}\Big),\nonumber\\
& \qquad \qquad \qquad \qquad \qquad \qquad \qquad \qquad i\in\cN^l,\\
&q_{i,t+1}=q_{i,t}-\epsilon\Big(\sum_{j\in\cM^v_t}\frac{\partial v_j}{\partial p_i} \frac{(v_{j,t}-\hat{v}_{j,t})}{\sigma^2_{v_j}}+\frac{\big(q_{i,t}-\hat{q}_{i,t} \big)^2}{\sigma^2_{q_i}}\Big),\nonumber\\
& \qquad \qquad \qquad \qquad \qquad \qquad \qquad \qquad i\in\cN^l,\label{eq:multiq}\\
&(p_{i,t+1},q_{i,t+1}) \in \{0,0\},\qquad \qquad \qquad \forall i\in\cN\backslash\cN^l, \label{eq:multip_vir}\\
&\bm{v}_{t+1}=\bm{g}(\bm{p}_{t+1},\bm{q}_{t+1}),\label{eq:v update}
\end{align}
\end{subequations}
where $\cM^v_t\subseteq \cM^v$ denotes the subset of nodes whose real-time measurements of voltage magnitude are received by the control center at time $t$.
There is no explicit expression of function $\bm{g}$. Hence, we estimate the voltage magnitude in \eqref{eq:v update} with the power flow solution calculated by OpenDSS given the updated power injection for each $t$.

As a specific case of \eqref{eq:z_update_LPF_SGD}, dynamics~\eqref{eq:SGD} follows the same convergence property as described in Theorem~\ref{the:convergence}.

\section{Simulation Results}\label{Simulation}
\subsection{Simulation Setup}
The numerical performance of the proposed online SGD algorithm is compared with the online GD algorithm \eqref{eq: grad_LPF+Feedback}, denoted by GD, and \eqref{eq: r-t grad_LPF+Feedback}, denoted by GO, on a unbalance three-phase IEEE-123 bus feeder. We use the realistic load and solar irridiance data from feeders in Anatolia, California on a day of August in 2012. The data is of an one-second resolution\cite{Bank13} from 6 a.m. to 6 p.m., amounting to 43,200 successive scenarios; see Fig.~\ref{fig:Real-time PV+Load Data} for the real power injection of PV and load profile of one node. Measurements are obtained by adding Gaussian noise of normal distribution to the true power flow solution solved via OpenDss. We randomly place voltage meters on $12\%$ of all the nodes. The standard deviations of the $29$ voltage measurements are $0.01$ p.u. All the load nodes have pseudo-measurements with $50\%$ error \cite{Yao2019, Marco13}. At each second, only one voltage magnitude and three pairs of power injection measurements are sent to the estimator. All the  simulations are conducted on an HP ENVY NOTEBOOK with processor Intel®Core (TM) i7-6500CPU 2.5 GHz, RAM 8.0 GB and 64-bit operating system, running Matlab R2016a on Windows 10 Home Version.

\begin{figure}[htbp]
\centerline{\includegraphics[width=.5\textwidth]{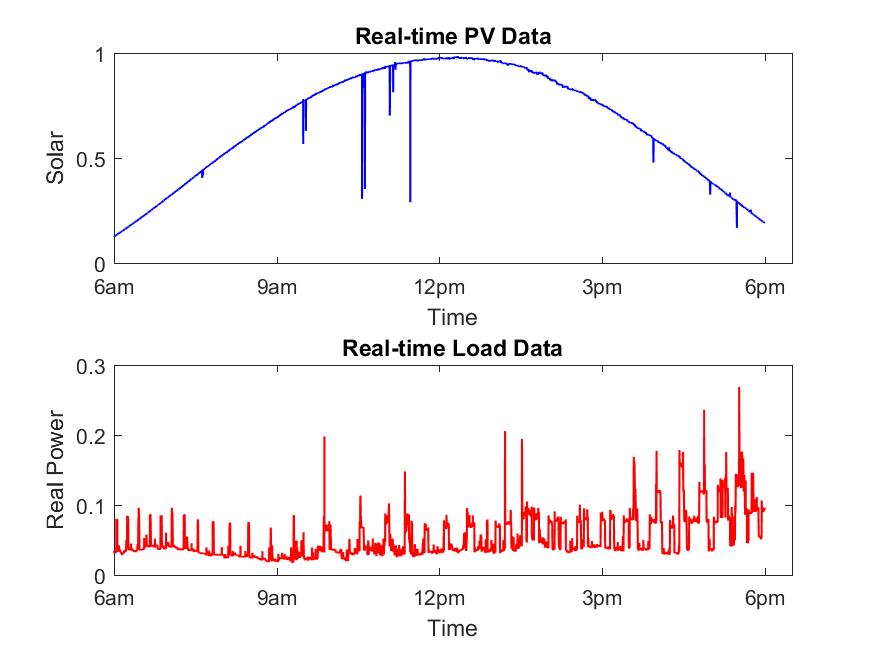}}
\caption{Real-time PV (upper) and load profile (lower) with 1 second temporal granularity.}
\label{fig:Real-time PV+Load Data}
\end{figure}

\subsection{Numerical Tests}
\subsubsection{Tracking the fast changing system states}
We present the true voltage magnitude of a randomly selected node and its estimation via the proposed SGD scheme and online GD algorithms at every second in Fig.~\ref{fig:Real-time Voltage Estimation}, where the blue curve represent the true voltage magnitude, while the yellow, purple and red dash curves represent the estimation voltages via GO, GD and SGD algorithms. In the plot, most of the blue curve is covered by the yellow curve, which implies that GO algorithm can track fast changing states very well. The accuracy is due to the complete access to the real-time measurements. The small mismatches between the blue and red curves validate that the SGD algorithm can track the voltage magnitudes accurately in real-time DSSE.

\begin{figure}[htbp]
\centerline{\includegraphics[width=.52\textwidth]{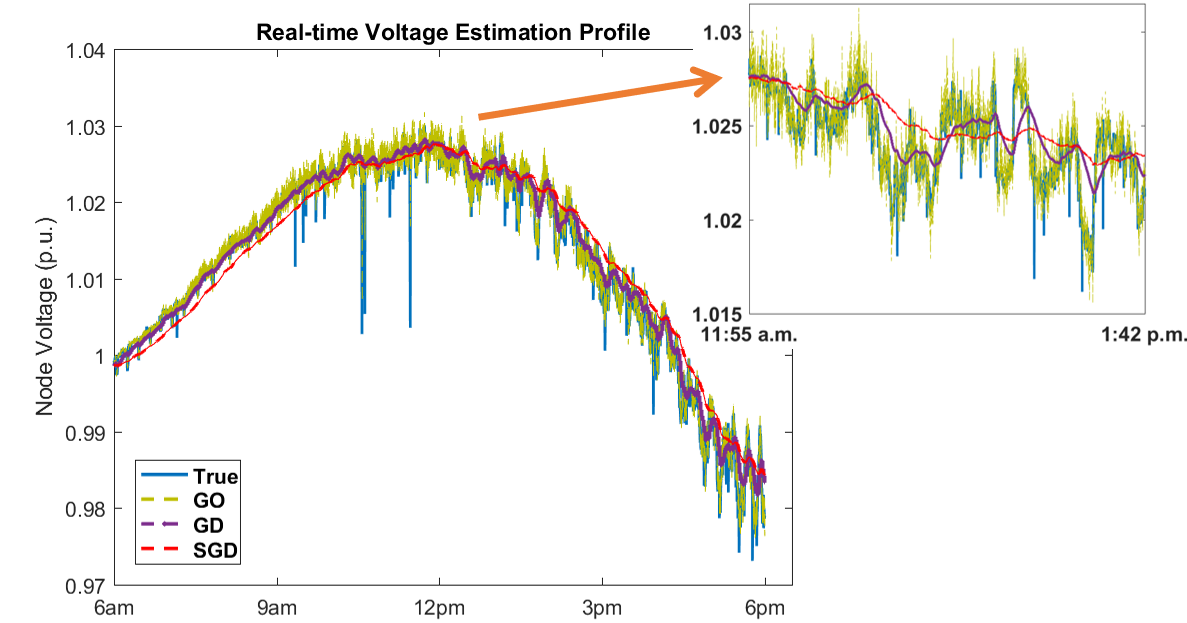}}
\caption{Comparison of the voltage estimation updated by the online stochastic gradient and gradient algorithms.}
\label{fig:Real-time Voltage Estimation}
\end{figure}

\subsubsection{Stability}\label{prove convergence}
The real-time average error of the voltage estimation is presented in Fig.~\ref{fig:running ave of Vm error}, where the blue, red and yellow lines represent the voltage estimation errors of SGD, GD and GO algorithms. The average error of voltage estimation for GO is less than that of SGD and GD. In a lot of cases, the averge error of voltage estimation of GD is smaller than SGD as it accesses to all the available real-time measurements at every second with the requirement of synchronizing all the measurement devices. However, the SGD does not need the requirement, and its average error is less than $1\%$ for most of the samples during the $12$-hour interval. 

\begin{figure}[htbp]
\centerline{\includegraphics[width=.5\textwidth]{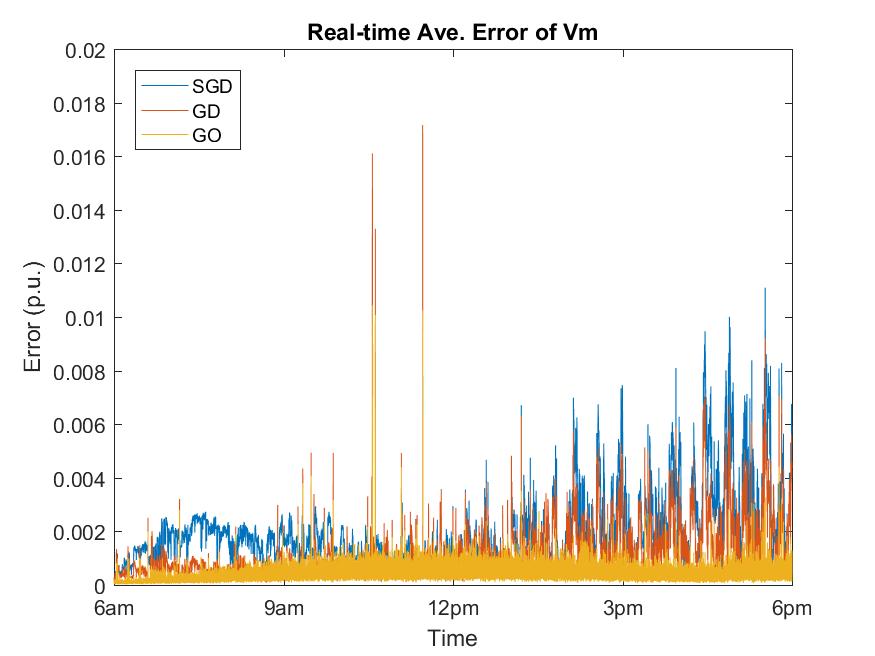}}
\caption{Real-time average error of the voltage estimation for SGD, GD and GO at every second.}
\label{fig:running ave of Vm error}
\end{figure}

\subsubsection{Accuracy and Computation Time}
In this subsection, we compare three aspects of the numerical results of SGD, GD and GO as follows: 1) the average computation time, 2) the average estimation error of the voltage magnitude per node, and 3) the average of the maximum estimation error of the voltage magnitude per sample. The comparison is presented in Table~\ref{Comparison btw SGD and GD}. The second row of the table presents the results of SGD. As we can see, the average error per node is $2 e^{-3}$ p.u., and the average maximal error per scenario is $5.3 e^{-3}$ p.u.. Both errors are less than $1\%$, and very close to the errors of GD and GO. The result shows that the proposed SGD-based DSSE algorithm achieves accurate voltage estimation. The average computation time of SGD per iteration is $1.8 e^{-3}$ seconds, much less than $1$ second. Therefore, the accuracy and efficiency of the proposed SGD-based algorithm enables DSSE to handle a system featured with fast changing states with the asynchronous measurements in real-time.

\begin{table}[!ht]
\begin{center}
\caption{Computation Time and Estimation Errors of voltage magnitudes.}
\begin{tabular}{c c c c}
\dtoprule
& {\textbf{Ave. Time (s)}}
& {\textbf{Ave. Error (p.u.)}} 
& {\textbf{Ave. Max. Error (p.u.)}} \\
\hline
\textbf{SGD}  & $1.80e-03$  & $3.0e-03$ & $6.3e-03$\\
\textbf{GD}  & $1.80e-03$  & $1.7e-03$ & $4.3e-03$\\
\textbf{GO}  & $0.17$  & $1.1e-03$ & $2.8e-03$\\
\dbottomrule
\end{tabular}
\label{Comparison btw SGD and GD}
\end{center}
\end{table}

\section{Conclusion}\label{conclusion}
We have proposed an SGD-based real-time DSSE algorithm to process data stream of asynchronous measurements. The convergence analysis of the proposed algorithm has been established. The algorithm has been tested in an unbalanced three-phase IEEE-123 bus system under realistic solar and load data to effectively and accurately track fast-changing system states with asynchronous measurements. 

\nocite{*} 
\bibliographystyle{IEEEtran}
\bibliography{pscc2022_fullpaper.bib}

\end{document}